\documentclass[12pt]{article}
\usepackage[english]{babel}
\usepackage{amsmath,amssymb,amsfonts,latexsym}
\usepackage{color}
\usepackage{graphicx}
\newtheorem{theorem}{Theorem}

\newtheorem{proposition}{Proposition}
\newtheorem{definition}{Definition}

\newtheorem{remark}{Remark}
\newenvironment{proof}{\noindent {\bf Proof.}}{\hfill$\Box$}
\DeclareMathOperator{\byd}{\raisebox{-.2ex}{$\overset{\text{\tiny def}}{=}$}} 

\allowdisplaybreaks[3]

\begin{document}

\title{Lie remarkable partial differential equations \\
characterized by Lie algebras \\
of point symmetries}

\author{Matteo Gorgone, Francesco Oliveri\\
\ \\
{\footnotesize Department MIFT, University of Messina}\\
{\footnotesize Viale F. Stagno d'Alcontres 31, 98166 Messina, Italy}\\
{\footnotesize mgorgone@unime.it; foliveri@unime.it}
}

\date{Published in \textit{J. Geom. Phys.} \textbf{144}, 314--323 (2019).}

\maketitle

\begin{abstract}
Within the framework of inverse Lie problem, we give some
non-trivial examples of coupled Lie remarkable equations, \textit{i.e.}, classes of
differential equations that are in correspondence with
their Lie point symmetries. In particular, we determine hierarchies of second order 
partial differential equations uniquely characterized by affine transformations 
of $\mathbb{R}^{n+m}$, and a system of two third order partial differential equations in two
independent variables uniquely determined by the Lie algebra of projective transformations of 
$\mathbb{R}^4$.
\end{abstract}


\begin{center}
\emph{Dedicated to Joseph Krasil'shchik on the occasion of his 70th birthday.}
\end{center}

\section{Introduction}
Lie point symmetries of ordinary as well as partial differential equations are
(finite or infinitesimal) transformations of the independent and dependent
variables with the property of sending solutions into solutions
\cite{Ovsiannikov,Ibragimov,Olver1,Olver2,Many1999,Baumann,BlumanAnco,%
BlumanCheviakovAnco}; they in turn induce transformations of the derivatives of the latter with 
respect to the former by requiring 
that the contact conditions are preserved.
The set of all infinitesimal point symmetries, characterized by their infinitesimal 
generators, of a given differential equation span a
vector space that has also the structure of a Lie algebra.

In Lie group analysis of differential equations, the \emph{direct} problem of finding the 
admitted symmetries of a given differential equation can
be complemented by a natural \emph{inverse} problem, namely, the problem of finding
the most general form of a differential equation (once the number of the independent and 
dependent variables and the maximal order of derivatives are fixed) admitting a given Lie algebra as subalgebra 
of its infinitesimal point symmetries.

There are different ways to consider this inverse problem.
A first contribution  aimed at characterizing all differential equations
admitting a given symmetry group was given in \cite{BlumanCole}.
In \cite{Rosenhausone}, the problem whether there exist non--trivial
differential equations in one--to--one correspondence with their Lie point symmetries has
been considered: the equation for the surfaces in $\mathbb{R}^3$ with vanishing Gaussian 
curvature has been recognized to be uniquely determined by its Lie point symmetries.

It is simple to check that many relevant differential equations can not be 
characterized by their Lie point symmetries since the 
dimension of the admitted Lie algebra is too small: KdV equation, Burgers' equation, 
Kepler's equations, \ldots are just some examples.
An interesting contribution has been given by Rosenhaus in 1982 \cite{Rosenhausone}, 
who posed the problem of the unique determination of a differential equation by its group; in 
\cite{Rosenhausone}, the author considered the projective algebra of $\mathbb{R}^3$ and its subalgebras, and was 
able to prove that the equation of vanishing Gaussian curvature of surfaces in $\mathbb{R}^3$ (which is a 
Monge--Amp\`ere--type equation) is uniquely determined by its Lie point symmetries. Moreover, in 
\cite{Rosenhaus1,Rosenhaus2}, it was considered the problem of finding the minimal subalgebra of the algebra 
of point symmetries of the equation of vanishing Gaussian curvature of surfaces in $\mathbb{R}^3$ which uniquely 
determines it.
Several authors studied the problem of finding, for a given equation, an extension of the algebra of point 
symmetries for which the equation at hand is determined. For instance, in \cite{Krause1994} such a construction 
is performed in the case of ordinary differential equations by considering a non--local extension of the algebra 
of point symmetries,
and the notion of \emph{complete symmetry group} has been introduced (see also 
\cite{Nucci1996,ALF2001,MyeniLeach2007}). Also, in 
\cite{Rosenhausone}, it is shown that the algebra of contact symmetries is needed to characterize the equation 
of minimal surfaces.

The general approach to the inverse problem of constructing differential equations by 
using abstract Lie algebras requires to classify all possible realizations of the 
considered Lie algebra as algebra of vector fields on the space of independent and dependent
variables. Then, looking for the differential invariants of the realization
under consideration \cite{FushYeho}, under
suitable hypotheses of regularity, the most general differential equation admitting a 
given Lie algebra as subalgebra of point symmetries is locally given as zeros of a set 
of smooth functions of the differential invariants.

The notion of Lie remarkable equation has been first introduced by one of the authors (F.O.)
in \cite{Oliveri-NM2005}, and used to analyze several classes of Monge--Amp\`ere equations \cite{Oliveri-RendPalermo2006,MOV-Wascom05}.
Then, in \cite{MOV-JMAA2007}, the class of \emph{Lie remarkable equations}, as determined by their Lie point 
symmetries, has been characterized within a geometrical framework, distinguishing them in weakly and strongly 
Lie remarkable equations; moreover, 
either necessary or sufficient conditions for a differential equation to be strongly or weakly 
Lie remarkable have been stated. 
Roughly speaking, strongly Lie remarkable equations are uniquely determined by their Lie
point symmetries in the whole jet space, whereas weakly Lie remarkable equations are 
equations which do not intersect other equations admitting the same symmetries.
Examples of Lie remarkable partial differential equations (Monge--Amp\`ere equations, 
minimal surface equations) have been also given. 
Moreover, in  \cite{MOV-TMP2007,MOV-SPT2007}, differential equations
uniquely determined by some relevant Lie algebras of vector fields in $\mathbb{R}^3$ have been characterized. 
Finally, in \cite{MOSV-JGP2014}, the case of ordinary differential equations uniquely 
characterized by their Lie point symmetries has been investigated.

In this paper, we aim to determine Lie remarkable systems of partial differential 
equations, with special reference to the algebra 
of affine and projective transformations of $\mathbb{R}^{n+m}$, where $n\ge 2$ is the 
number of independent variables, and $m\ge 1$ the number of dependent ones. 

The plan of the paper is as follows. 
In Section~\ref{sec2}, we introduce a differential equation of order $r$ as a submanifold
of a suitable jet space \cite{Olver1,Many1999} of order $r$.  The symmetries are interpreted as particular vector fields on the jet 
space that are tangent to the differential equation.  Then, we review the definitions of strongly 
and weakly Lie remarkable equations, and recall the main results 
obtained in \cite{MOV-JMAA2007}. The (strong or weak) Lie remarkability of differential 
equations is proved by computing the rank of the distributions determined by the prolongations of a Lie algebra and determining the submanifolds where the rank lowers. In Section~\ref{sec3}, we identify a hierarchy of 
second order partial differential equations uniquely determined by the Lie algebra of affine 
transformations, whereas in Section~\ref{sec4}, we present third order strongly Lie 
remarkable equations characterized by the Lie algebra of projective transformations. 
All the required computations have been carried out with the help of the Reduce \cite{reduce} program ReLie \cite{Oliveri-Relie}.
Finally, Section~\ref{sec5} contains some concluding remarks.

\section{Theoretical setting}
\label{sec2}
In this Section, we recall some basic facts regarding Lie remarkable
equations (see \cite{MOV-JMAA2007} for further details). The theory is
carried out in the geometric framework of jet bundles \cite{Olver1,Many1999}, assuming all manifolds and maps to be $C^\infty$.

Let $E$ be an $(n+m)$--dimensional smooth manifold.  By using local charts
of the form $(x_i,u_\alpha)$, $i=1,\ldots, n$, $\alpha=1\ldots m$, we 
describe (locally) an $n$--dimensional submanifold $L\subset E$ as the
graph of a vector function $u_\alpha=f_\alpha(x_i)$. In what follows,
Latin indices run from $1$ to $n$, and Greek indices run from $1$ to $m$, unless
otherwise specified. We denote by $\chi(E)$ the Lie algebra of vector fields on $E$.

The \emph{$r$--jet of $n$--dimensional submanifolds of $E$} (also known
as extended jet bundles \cite{Olver1}, or manifold of contact elements),
$J^r(E,n)$, is the set of equivalence classes of submanifolds having at 
$p\in E$ a contact of order $r$. It has a smooth manifold structure with local charts
$(x_i,u_\alpha,u_{\alpha,j_1\ldots j_k})$, where 
\[
u_{\alpha,j_1\ldots j_k}\equiv\frac{\partial^{k}u_\alpha}{\partial x_{j_1}\ldots\partial 
x_{j_k}}=\frac{\partial^{k}f_\alpha}{\partial x_{j_1}\ldots\partial 
x_{j_k}}, \qquad 1\le k \le r. 
\]
It is  
\[
\dim(J^r(E,n)) = n+m\binom{n+r}{r}.
\]
On $J^r(E,n)$ we have the contact distribution generated by the total
derivative (Lie derivative)
\[
\frac{D}{Dx_{i}}\byd \frac{\partial}{\partial x_i}+ 
\sum_{\alpha=1}^m \left(u_{\alpha, i}\frac{\partial}{\partial u_{\alpha}}
+\sum_{k=1}^{r-1}\sum_{j_1=1}^n\ldots\sum_{j_k=1}^n u_{\alpha,j_1\ldots j_k  i}\frac{\partial}
{\partial u_{\alpha,j_1\ldots j_k}}\right).
\]
Any vector field $\Xi\in\chi(E)$, 
\[
\Xi=\sum_{i=1}^n\xi_i(x_j,u_\beta)\frac{\partial}{\partial x_i}+
\sum_{\alpha=1}^m\eta_\alpha(x_j,u_\beta)\frac{\partial}{\partial u_\alpha},
\]
can be lifted to a vector field $\Xi^{(r)}\in\chi(J^r(E,n))$, 
\[
\Xi^{(r)}=\sum_{i=1}^n\xi_i\frac{\partial}{\partial x_i}+\sum_{\alpha=1}^m\left(\eta_\alpha
\frac{\partial}{\partial u_\alpha}+\sum_{k=1}^{r}\sum_{j_1=1}^n\ldots\sum_{j_k=1}^n \eta_{[\alpha,j_1\ldots j_k]}\frac{\partial}
{\partial u_{\alpha,j_1\ldots j_k}}\right),
\]
where
\[
\eta_{[\alpha,j_1\ldots j_k]}=\frac{D\eta_{[\alpha,j_1\ldots j_{k-1}]}}{Dx_{j_k}}
-\sum_{\ell=1}^n\frac{D\xi_\ell}{Dx_{j_k}}u_{\alpha,j_1\ldots j_{k-1}\ell}, \qquad
\eta_{[\alpha]}=\eta_\alpha.
\]

A \emph{differential equation $\mathcal{E}$ of order $r$ on an $n$--dimensional
submanifold of a manifold $E$} is a submanifold of $J^r(E,n)$. An
\emph{infinitesimal point symmetry} of $\mathcal{E}$ is a vector field of the type
$\Xi^{(r)}$ which is tangent to $\mathcal{E}$. If $\mathcal{E}$ is locally described by
\[
\Delta_i=0,\qquad i=1,\ldots, q<\dim \left(J^r(E,n)\right), 
\]
where 
\[
\Delta_i: J^r(E,n)\rightarrow\mathbb{R},
\]
then point symmetries are the solutions of the system 
\[
\left.\Xi^{(r)}\left(\Delta_i\right)\right|_{\Delta_1=\cdots=\Delta_q = 0}=0,  \quad i=1,
\ldots,q.
\]
We denote by $\hbox{sym}(\mathcal{E})$ the Lie algebra of infinitesimal point symmetries 
of the equation $\mathcal{E}$.

The problem of determining the Lie algebra $\hbox{sym}(\mathcal{E})$ is said to be the direct Lie problem. Conversely, given a Lie subalgebra $\mathfrak{S}\subset \chi(J^r(E,n))$ of contact vector fields, we consider the inverse Lie problem, \emph{i.e.}, the problem of classifying the equations $\mathcal{E}\subset J^r (E, n)$ such that 
$\hbox{sym}(\mathcal{E})\supseteq \mathfrak{S}$.

For the reader's convenience, let us now recall the definitions and the main properties, 
contained in \cite{MOV-JMAA2007}, of differential equations which are characterized by their 
Lie point symmetries, that we call \emph{Lie remarkable}.

\begin{definition}[Weakly Lie remarkable equations]
\label{def:weak}
Let $E$ be a manifold, $\dim(E)=n+m$, and let $r\in\mathbb{N}$, $r>0$.  An
$\ell$--dimensional equation $\mathcal{E}\subset J^r(E,n)$ is said to be
\emph{weakly Lie remarkable} with respect to a Lie algebra $\mathfrak{S}$ of point 
symmetries if $\mathcal{E}$ is the only maximal (with respect
to the inclusion) $\ell$--dimensional equation in $J^r(E,n)$ passing at any
$\theta\in\mathcal{E}$ and admitting $\mathfrak{S}$  as  a subalgebra of its 
infinitesimal point symmetries.
\end{definition}

\begin{definition}[Strongly Lie remarkable equations]
\label{def:strong}
Let $E$ be a manifold, $\dim \left(E\right)=n+m$, and let $r\in\mathbb{N}$, $r>0$.  An
$\ell$--dimensional equation $\mathcal{E}\subset J^r(E,n)$ is said to be
\emph{strongly Lie remarkable} with respect to a Lie algebra $\mathfrak{S}$ of point 
symmetries if $\mathcal{E}$ is the only maximal (with
respect to the inclusion) $\ell$--dimensional equation in $J^r(E,n)$
admitting $\mathfrak{S}$  as a subalgebra of its infinitesimal point symmetries.
\end{definition}

Of course, a strongly Lie remarkable equation is also weakly Lie remarkable.

Some direct consequences of the above definitions are in order.  For each 
$\theta\in J^r(E,n)$, let us denote by $S_\theta(\mathcal{E})\subset T_\theta J^r(E,n)$ 
the subspace generated by the values of infinitesimal point symmetries of $\mathcal{E}$ 
at $\theta$.
Let us set 
\[
S(\mathcal{E})\byd \bigcup_{\theta\in J^r(E,n)} S_\theta(\mathcal{E}). 
\]  
In general, $\dim\left(S_\theta(\mathcal{E})\right)$ may change with 
$\theta\in J^r(E,n)$.  It is
clear that $\dim\left(\hbox{sym}(\mathcal{E})\right) \geq \dim\left(S_\theta(\mathcal{E})\right)$, 
for all  $\theta\in J^r(E,n)$. If
the rank of $S(\mathcal{E})$ at each $\theta\in J^r(E,n)$ equals 
$\dim\left(\hbox{sym}(\mathcal{E})\right)$, then
$S(\mathcal{E})$ is an involutive (smooth) distribution.  A submanifold $N$ of $J^r (E, n)$ 
is an integral submanifold of $S(\mathcal{E})$ if $T_\theta N = S_\theta (\mathcal{E})$ for 
each $\theta\in N$. Of course, an integral submanifold of $S(\mathcal{E})$ is an equation in 
$J^r(E,n)$ which admits all elements in $\hbox{sym}(\mathcal{E})$ as infinitesimal point 
symmetries. Moreover, due to the fact that point symmetries of $\mathcal{E}$ are tangent to 
$\mathcal{E}$, we have $\dim(S_\theta(\mathcal{E}))\le \ell$.

The points of $J^r(E,n)$ of maximal rank of $S(\mathcal{E})$ form an open set of $J^r(E,n)$
(see \cite{MOV-JMAA2007}). It follows that $\mathcal{E}$ can not coincide with 
the set of points of maximal rank of $S(\mathcal{E})$. 

In \cite{MOV-JMAA2007}, the following 
results, concerned with either necessary or sufficient conditions for a differential 
equation to be weakly or strongly Lie remarkable, have been proved. 

\begin{theorem}
\label{th:weaknecessary}
A necessary condition for the differential equation $\mathcal{E}$ to be
strongly Lie remarkable is that $\dim \left(\hbox{sym}(\mathcal{E})\right)> \dim
\left(\mathcal{E}\right)$.
\end{theorem}

\begin{theorem}
\label{th:strongnecessary}
A necessary condition for the differential equation $\mathcal{E}$ to be weakly
Lie remarkable is that $\dim \left(\hbox{sym}(\mathcal{E})\right)\geq \dim
\left(\mathcal{E}\right)$.
\end{theorem}  

\begin{theorem} 
\label{th:weaksufficient}
If $S(\mathcal{E})|_{\mathcal{E}}$ is an $\ell$--dimensional distribution on 
$\mathcal{E}\subset J^r(E,n)$, then $\mathcal{E}$ is a weakly Lie remarkable equation.
\end{theorem}

\begin{theorem}
\label{th:strongsufficient}
Let $S(\mathcal{E})$ be such that for any $\theta\not\in\mathcal{E}$ we have 
$\dim \left(S_\theta(\mathcal{E})\right)> \ell$. Then $\mathcal{E}$ is a strongly Lie 
remarkable equation.
\end{theorem}

By using these definitions and results, in the next sections we will determine some classes of new Lie remarkable differential equations.

\section{Second order Lie remarkable equations}
\label{sec3}
The Lie algebra $\mathcal{A}(\mathbb{R}^{n+m})$ of affine transformations of 
$\mathbb{R}^{n+m}$  is spanned by the vector fields
\[
\frac{\partial}{\partial a}, \qquad a\frac{\partial}{\partial b},
\]
where $a,b \in \{x_1,\ldots,x_n,u_1,\ldots,u_m\}$; it is
\[
\dim(\mathcal{A}(\mathbb{R}^{n+m})) =(m+n)(m+n+1);
\]
on the contrary, as far as the second order jet space is concerned, it is 
\[
\dim(J^2(\mathbb{R}^{n+m},n))=n+m(n+1)(n+2)/2,
\]
whereupon it is immediate to prove that 
\[
\dim(\mathcal{A}(\mathbb{R}^{n+m}))  >\dim(J^2(\mathbb{R}^{n+m},n))
\] 
for all $m$ provided that $n\le 5$. Nevertheless, in order to characterize 
strongly (weakly, respectively) Lie remarkable equations, it is necessary that the dimension of 
the Lie algebra  be greater (greater or equal, respectively) to the dimension of the 
equation.  In the following, we shall prove that strongly Lie remarkable equations corresponding to the Lie algebra of affine transformations of $\mathbb{R}^{n+m}$ 
exist for all values of $m$ and $n$.

In \cite{MOV-JMAA2007}, it was proved that for $n=2$, $m=1$, the Monge--Amp\`ere equation 
for a surface with vanishing Gaussian curvature,
\begin{equation}
\label{MA2}
u_{,11}u_{,22}-u_{,12}^2=0,
\end{equation}
is strongly Lie remarkable: equation~\eqref{MA2} is uniquely characterized by the Lie algebra 
of  affine transformations, even if it admits the Lie algebra of projective transformations 
of $\mathbb{R}^3$.

This result has been extended \cite{MOV-JMAA2007} to $n$--dimensional manifolds in 
$\mathbb{R}^{n+1}$, as stated by the next theorem. 

\begin{theorem}
The unique second order partial differential equation for the unknown $u(x_1,\ldots,x_n)$ 
$(n\ge 2)$,  characterized by the Lie algebra of affine transformations of 
$\mathbb{R}^{n+1}$, is
\begin{equation}
\label{detHessiano_n}
\det \left\|
\begin{array}
{llll}
u_{,11} & u_{,12} & \ldots & u_{,1n}\\
u_{,12} & u_{,22} & \ldots & u_{,2n}\\
\ldots & \ldots & \ldots & \ldots\\
u_{,1n} & u_{,2n} & \ldots & u_{,nn}\\
\end{array}
\right\|=0,
\end{equation}
that so is strongly Lie remarkable.
\end{theorem}
\begin{proof}
The second order prolongation of vector fields of Lie algebra of affine transformations of 
$\mathbb{R}^{n+1}$ gives a distribution of maximal rank (equal to $(n+1)(n+2)$); the only 
maximal manifold where this rank lowers to $(n+1)(n+2)-1$ is characterized by the local 
expression \eqref{detHessiano_n}. 
\end{proof}

The Monge--Amp\`ere equation for a surface with vanishing Gaussian curvature is a particular instance of the equation introduced by Amp\`ere in a paper published in 1815 on the \emph{Journal de l'\'Ecole Polytechnique} \cite{Ampere}, having the form
\begin{equation}
\label{OriginalMA}
H \frac{\partial^2 u}{\partial t^2} + 2K \frac{\partial^2 u}{\partial t\partial x} + L \frac{\partial^2 u}{\partial x^2} + M + N\left(\frac{\partial^2 u}{\partial t^2}\frac{\partial^2 u}{\partial x^2}-\left(\frac{\partial^2 u}{\partial t\partial x}\right)^2\right) = 0,
\end{equation}
where the coefficients $H$, $K$, $L$, $M$, $N \neq 0$ depend on $t$, $x$, $u$ and first order derivatives.
A remarkable property of equation~\eqref{OriginalMA} has been proved by Boillat \cite{Boillat1968}, who showed that \eqref{OriginalMA} is the most general second order nonlinear hyperbolic equation in $(1+1)$ dimensions that is completely exceptional. In 1973, Ruggeri  \cite{Ruggeri}, using the property of complete exceptionality, generalized the Monge--Amp\`ere equation to $(2+ 1)$ dimensions. Moreover, Donato, Ramgulam, and Rogers 
\cite{DonatoRamgulamRogers} characterized the $(3+1)$--dimensional Monge--Amp\`ere equation. Boillat \cite{Boillat_n}, with the same procedure, derived the general form of 
Monge--Amp\`ere equation  for a function $u(x_1,\ldots,x_n)$. All these second order Monge--Amp\`ere
equations are given as a linear combination (with coefficients arbitrary functions of the independent variables, 
the dependent variable and its first order derivatives) of all minors of all orders extracted from the Hessian 
matrix. Some of these equations can be reduced to quasilinear or linear form
\cite{Oliveri_MAlinearizable,GorgoneOliveri1,GorgoneOliveri2}. 
The property of complete exceptionality has been used by Boillat \cite{Boillat_higher} 
to characterize higher order Monge--Amp\`ere equations for a function $u(x_1,x_2)$; in 
such a case, they are given as a linear combination of all minors of all orders 
extracted from the Hankel matrix.

Now, we consider the situation where $m>1$, \emph{i.e.}, the case of a system of partial differential equations, 
and derive new strongly Lie remarkable equations.

\begin{theorem}
The system of second order partial differential equations locally described in $J^2(\mathbb{R}^4,2)$ by
\begin{equation}
\label{strong222}
\left\{
\begin{aligned}
&u_{1,11}u_{2,12}-u_{1,12}u_{2,11}=0,\\
&u_{1,11}u_{2,22}-u_{1,22}u_{2,11}=0,
\end{aligned}
\right.
\end{equation}
is strongly Lie remarkable.
\end{theorem}
\begin{proof}
System \eqref{strong222} characterizes a 12--dimensional manifold in the second order jet 
space $J^2(\mathbb{R}^4,2)$, that has dimension 14. The dimension of the Lie algebra of 
affine transformations of $\mathbb{R}^4$ is 20, and, by a straightforward calculation, 
one can easily recognize that the rank of the distribution of its second order 
prolongation is 14, reducing  to 12 only on the 12--dimensional manifold of 
$J^2(\mathbb{R}^4,2)$ characterized by \eqref{strong222};
therefore, system \eqref{strong222} is strongly Lie remarkable. In passing, we note that 
system~\eqref{strong222} admits also the symmetries generated by
\begin{equation}
a\left(x_1\frac{\partial}{\partial x_1}+x_2\frac{\partial}{\partial x_2}
+u_1\frac{\partial}{\partial u_1}+
u_2\frac{\partial}{\partial u_2}\right),
\end{equation}
where $a\in\{x_1,x_2,u_1,u_2\}$, \emph{i.e.}, it admits the 24--dimensional Lie algebra 
of projective transformations of $\mathbb{R}^4$.
\end{proof}

\begin{remark}
In a 4--dimensional Euclidean space, the functions
\begin{equation}
\label{M2}
u_1=u_1(x_1,x_2),\quad u_2=u_2(x_1,x_2) 
\end{equation}
characterize a 2--dimensional surface. In \cite{Ganchev-Milousheva}, the authors introduced an invariant linear map of Weingarten type in the tangent space of the surface, which generates two invariants $k$ and $\kappa$.
These invariants divide the points of the surface into 4 types: flat, elliptic, parabolic and hyperbolic.
The two invariants (see \cite{Ganchev-Milousheva} for the details) for the surface defined by \eqref{M2} have the form
\begin{equation*}
\begin{aligned}
k &=\frac{4(u_{1,11}u_{2,12}-u_{1,12}u_{2,11})(u_{1,12}u_{2,22}-u_{1,22}u_{2,12})-(u_{1,11}u_{2,22}-u_{1,22}u_{2,11})^2}{(1+u_{1,1}^2+u_{1,2}^2+u_{2,1}^2+u_{2,2}^2+(u_{1,1}u_{2,2}+u_{1,2}u_{2,1})^2)^3},\\
\kappa&=\frac{(1+u_{1,1}^2+u_{2,1}^2)(u_{1,12}u_{2,22}-u_{1,22}u_{2,12})}{(1+u_{1,1}^2+u_{1,2}^2+u_{2,1}^2+u_{2,2}^2+(u_{1,1}u_{2,2}+u_{1,2}u_{2,1})^2)^2}\\
&+\frac{(1+u_{1,2}^2+u_{2,2}^2)(u_{1,11}u_{2,12}-u_{1,12}u_{2,11})}{(1+u_{1,1}^2+u_{1,2}^2+u_{2,1}^2+u_{2,2}^2+(u_{1,1}u_{2,2}+u_{1,2}u_{2,1})^2)^2}\\
&-\frac{(u_{1,1}u_{1,2}-u_{2,1}u_{2,2})(u_{1,11}u_{2,22}-u_{1,22}u_{2,11})}{(1+u_{1,1}^2+u_{1,2}^2+u_{2,1}^2+u_{2,2}^2+(u_{1,1}u_{2,2}+u_{1,2}u_{2,1})^2)^2}.
\end{aligned}
\end{equation*}

It is immediately seen that both invariants vanish identically on the equations \eqref{strong222}, 
whereupon it remains proved that these equations characterize the surfaces consisting of flat points:
such surfaces are either planar surfaces or developable ruled surfaces \cite{Ganchev-Milousheva}.
\end{remark}

Second order strongly Lie remarkable systems do exist also for $m>2$ or $n>2$. 
Before stating the theorem that 
identifies such systems for arbitrary values of $m$ 
and $n$, let us consider some special cases.

\begin{theorem}
The system of second order partial differential equations   locally described in $J^2(\mathbb{R}^5,3)$ by
\begin{equation}
\label{strong322}
\left\{
\begin{aligned}
&u_{1,11}u_{2,12}-u_{1,12}u_{2,11}=0,\\
&u_{1,11}u_{2,13}-u_{1,13}u_{2,11}=0,\\
&u_{1,11}u_{2,22}-u_{1,22}u_{2,11}=0,\\
&u_{1,11}u_{2,23}-u_{1,23}u_{2,11}=0,\\
&u_{1,11}u_{2,33}-u_{1,33}u_{2,11}=0,
\end{aligned}
\right.
\end{equation}
is strongly Lie remarkable.
\end{theorem}
\begin{proof}
System \eqref{strong322} characterizes a 18--dimensional manifold in the second order jet 
space $J^2(\mathbb{R}^5,3)$, that has dimension 23. The dimension of the Lie algebra of 
affine transformations of $\mathbb{R}^5$ is 30, and, by a straightforward calculation, 
one can easily recognize that the rank of the distribution of its second order 
prolongation is 23, reducing  to 18 only on the 18--dimensional manifold of 
$J^2(\mathbb{R}^5,3)$ characterized by \eqref{strong322},
that, consequently, is strongly Lie remarkable. Actually,  the 
system~\eqref{strong322} admits the 
35--dimensional Lie algebra of projective transformations of $\mathbb{R}^5$.
\end{proof}

\begin{theorem}
The system of second order partial differential equations  locally described in $J^2(\mathbb{R}^5,2)$ by
\begin{equation}
\label{strong232}
\left\{
\begin{aligned}
&u_{1,11}u_{2,12}-u_{1,12}u_{2,11}=0,\\
&u_{1,11}u_{2,22}-u_{1,22}u_{2,11}=0,\\
&u_{2,11}u_{3,12}-u_{2,12}u_{3,11}=0,\\
&u_{2,11}u_{3,22}-u_{2,22}u_{3,11}=0,
\end{aligned}
\right.
\end{equation}
is strongly Lie remarkable.
\end{theorem}
\begin{proof}
System \eqref{strong232} characterizes a 16--dimensional manifold in the second order jet 
space $J^2(\mathbb{R}^5,2)$, that has dimension 20. The dimension of the Lie algebra of 
affine transformations of $\mathbb{R}^5$ is 30, and, by a straightforward calculation, 
one can easily recognize that the rank of the distribution of its second order 
prolongation is 20, reducing  to 16 only on the 16--dimensional manifold of 
$J^2(\mathbb{R}^5,2)$ characterized by \eqref{strong232},
that, consequently, is strongly Lie remarkable. Actually, it is easily verified  that 
system~\eqref{strong232} 
admits the 35--dimensional Lie algebra of projective transformations of $\mathbb{R}^5$ too.
\end{proof}

For arbitrary values of $m$ and $n$, we are able to state the following result.

\begin{theorem}
The system of second order partial differential equations in the unknowns
$(u_1(x_1,\ldots,x_n),\ldots,u_m(x_1,\ldots,x_n))$
\begin{equation}
\label{strongnm2}
\boldsymbol\Delta=\mathbf{0}, 
\end{equation}
with the components of $\boldsymbol\Delta$ given by
\begin{equation}
\label{strongnm2components}
\Delta_{\alpha,\alpha+1;1,1,p,q}\equiv 
u_{\alpha,11} u_{\alpha+1,pq}-u_{\alpha,pq} u_{\alpha+1,11}=0,
\end{equation}
where $\alpha=1,\ldots,m-1$ and $p,q=1,\ldots, n$, is strongly Lie remarkable.
\end{theorem}
\begin{proof}
System~\eqref{strongnm2}, made by $(m-1)(n^2+n-2)/2$ independent differential equations, 
characterizes a manifold with dimension $d$ in the second order jet space 
$J^2(\mathbb{R}^{n+m},n)$, where
\[
d = n^2+2n+3m-2+\frac{mn(1-n)}{2}.
\]
A simple analysis shows that $d<(n+m)(n+m+1)$ for all positive values of $m$ and $n$, so 
that the number of symmetries in the affine Lie algebra of $\mathbb{R}^{n+m}$, according 
to Theorem~\ref{th:strongsufficient}, may suffice. Moreover, the rank of the distribution of the 
second order prolongation of the generators of affine Lie algebra in $\mathbb{R}^{n+m}$ is 
maximal (the minimum between the dimension of affine Lie algebra and the dimension of jet 
space).  

It is easy to recognize that from \eqref{strongnm2components} it follows
\begin{equation}
\Delta_{\alpha,\beta;i,j,p,q} \equiv 
u_{\alpha,ij} u_{\beta,pq}-u_{\alpha,pq} u_{\beta+1,ij}=0,
\end{equation}
for all $\alpha,\beta \in \{1,\ldots,m\}$ and all $i,j,p,q \in\{1,\ldots,n\}$.

Now, let us prove that system~\eqref{strongnm2} admits as subalgebra of the algebra of 
its Lie point symmetries the Lie algebra of affine transformations. In the sequel, we use  the 
Einstein summation convention over repeated indices.

Consider the infinitesimal generators of the affine Lie algebra in 
$\mathbb{R}^{n+m}$, say
\begin{equation}\label{affine_gen}
\frac{\partial}{\partial x_i}, \qquad \frac{\partial}{\partial u_\beta}, \qquad x_j 
\frac{\partial}{\partial x_i},\qquad 
u_\beta \frac{\partial}{\partial x_i},\qquad x_i \frac{\partial}{\partial u_\beta},\qquad 
u_\gamma \frac{\partial}{\partial u_\beta},
\end{equation}
where $i,j=1,\ldots,n$ and $\beta,\gamma=1,\ldots,m$.
We will prove that each of the vector fields in \eqref{affine_gen} is admitted by 
system~\eqref{strongnm2}.

Since neither the independent variables nor the dependent ones  appear explicitly in 
system~\eqref{strongnm2}, it is 
evident the invariance with respect to the $(n+m)$ translations of the independent and 
dependent variables.
Let us consider the remaining vector fields.
\begin{enumerate}

\item $\Xi=x_j \frac{\partial}{\partial x_i}$. 
Its second order prolongation is
\begin{equation}
\Xi^{(2)}=x_j \frac{\partial}{\partial x_i}
-\delta_{j,k}u_{\gamma,i}\frac{\partial}{\partial u_{\gamma,k}}-\left(\delta_{j,\ell}
u_{\gamma,ik} 
+\delta_{j,k} u_{\gamma,i\ell}\right)\frac{\partial}{\partial u_{\gamma,k\ell}},
\end{equation}
where $\delta_{i,j}$ is the Kronecker symbol.
It is:
\begin{equation}
\begin{aligned}
&\Xi ^{(2)}\left(\Delta_{\alpha,\alpha+1;1,1,p,q}\right)=\\ 
&\;=2\delta_{j,1} \left(u_{\alpha,1i}u_{\alpha+1,pq}-
u_{\alpha,pq}u_{\alpha+1,1i}\right)\\
&\;+\delta_{j,p}\left(u_{\alpha,11}u_{\alpha+1,iq}-u_{\alpha,iq}u_{\alpha+1,11}\right)\\
&\;+\delta_{j,q}\left(u_{\alpha,11}u_{\alpha+1,ip}-u_{\alpha,ip}u_{\alpha+1,11}\right)\\
&\;=2\delta_{j,1}\Delta_{\alpha,\alpha+1;1,i,p,q}+\delta_{j,p}\Delta_{\alpha,\alpha+1;1,1,i,q}
+\delta_{j,q}\Delta_{\alpha,\alpha+1;1,1,i,p},
\end{aligned}
\end{equation}
that vanishes on $\boldsymbol\Delta=\mathbf{0}$;

\item  $\Xi=u_\beta \frac{\partial}{\partial x_i}$. Its second order prolongation is
\begin{equation}
\begin{aligned}
\Xi^{(2)}&=u_\beta \frac{\partial}{\partial x_i}
-u_{\gamma,i}u_{\beta,k}\frac{\partial}{\partial u_{\gamma,k}}\\
&-\left(u_{\gamma,i}u_{\beta,k\ell}
+u_{\beta,k}u_{\gamma,i\ell}
+u_{\beta,\ell}u_{\gamma,ik}\right)\frac{\partial}{\partial u_{\gamma,k\ell}},
\end{aligned}
\end{equation}
whereupon 
\begin{equation}
\begin{aligned}
&\Xi^{(2)}\left(\Delta_{\alpha,\alpha+1;1,1,p,q}\right)=\\
&\;=u_{\alpha,i}\left(u_{\beta,11}u_{\alpha+1,pq}-u_{\beta,pq}u_{\alpha+1,11}\right)\\
&\;+u_{\alpha+1,i}\left(u_{\alpha,11}u_{\beta,pq}-u_{\alpha,pq}u_{\beta,11}\right)\\
&\;+2u_{\beta,1}\left(u_{\alpha,1i}u_{\alpha+1,pq}-u_{\alpha,pq}u_{\alpha+1,1i}\right)\\
&\;+u_{\beta,p}\left(u_{\alpha,11}u_{\alpha+1,iq}-u_{\alpha,iq}u_{\alpha+1,11}\right)\\
&\;+u_{\beta,q}\left(u_{\alpha,11}u_{\alpha+1,ip}-u_{\alpha,ip}u_{\alpha+1,11}\right)\\
&\;=u_{\alpha,i}\Delta_{\beta,\alpha+1;1,1,p,q}
+u_{\alpha+1,i}\Delta_{\alpha,\beta;1,1,p,q}\\
&\;+2u_{\beta,1}\Delta_{\alpha,\alpha+1;1,i,p,q}
+u_{\beta,p}\Delta_{\alpha,\alpha+1;1,1,i,q}
+u_{\beta,q}\Delta_{\alpha,\alpha+1;1,1,i,p},
\end{aligned}
\end{equation}
vanishing on $\boldsymbol\Delta=\mathbf{0}$;

\item  $\Xi=x_i \frac{\partial}{\partial u_\beta}$. Its second order prolongation reads
\begin{equation}
\Xi^{(2)}=x_i \frac{\partial}{\partial u_\beta}+\delta_{\beta,\gamma}\delta_{i ,k}
\frac{\partial}{\partial u_{\gamma,k}},
\end{equation}
whereupon
\begin{equation}
\Xi^{(2)}\left(\Delta_{\alpha,\alpha+1;1,1,p,q}\right)=0;
\end{equation}

\item $\Xi=u_\gamma \frac{\partial}{\partial u_\beta}$. Its second order prolongation is
\begin{equation}
\Xi^{(2)}=u_\gamma \frac{\partial}{\partial u_\beta}+\delta_{\beta,\mu}u_{\gamma,k}
\frac{\partial}{\partial u_{\mu,k}}
+\delta_{\beta,\mu}u_{\gamma,k\ell}
\frac{\partial}{\partial u_{\mu,k\ell}},
\end{equation}
whereupon
\begin{equation}
\begin{aligned}
&\Xi^{(2)}\left(\Delta_{\alpha,\alpha+1;1,1,p,q}\right)=\\
&\;=\delta_{\alpha ,\beta}\left(u_{\gamma,11}u_{\alpha+1,pq}-u_{\gamma,pq}u_{\alpha+1,11}
\right)\\
&\;+\delta_{\alpha+1, \beta}
\left(u_{\alpha,11}u_{\gamma,pq}-u_{\alpha,pq}u_{\gamma,11}\right)\\
&\;=\delta_{\alpha, \beta}\Delta_{\gamma,\alpha+1;1,1,p,q}+\delta_{\alpha+1,\beta}
\Delta_{\alpha,\gamma;1,1,p,q},
\end{aligned}
\end{equation}
vanishing on $\boldsymbol\Delta=\mathbf{0}$, so completing the proof.
\end{enumerate}
\end{proof}

We also observe that system \eqref{strongnm2} is invariant with respect to the projective transformations 
of $\mathbb{R}^{n+m}$ too.
\begin{proposition}
System~\eqref{strongnm2} admits the symmetries generated 
by the vector fields
\begin{equation}
a\left(\sum_{i=1}^n x_i\frac{\partial}{\partial x_i}+\sum_{\alpha=1}^m u_\alpha
\frac{\partial}{\partial u_\alpha}\right)
\end{equation}
with $a\in\{x_1,\ldots,x_n,u_1,\ldots,u_m\}$,
so that system~\eqref{strongnm2} is strongly Lie remarkable also with respect to the Lie 
algebra of projective transformations of $\mathbb{R}^{n+m}$.
\end{proposition}  
\begin{proof}
We need to prove that system \eqref{strongnm2} admits also the symmetries spanned by the $(n+m)$
vector fields
\begin{equation}
x_j\left(\sum_{i=1}^nx_i\frac{\partial}{\partial x_i}+\sum_{\gamma=1}^m u_\gamma
\frac{\partial}{\partial u_\gamma}\right),\qquad u_\beta\left(\sum_{i=1}^nx_i
\frac{\partial}{\partial x_i}+\sum_{\gamma=1}^mu_\gamma\frac{\partial}{\partial u_\gamma}
\right),
\end{equation} 
where $j=1,\ldots,n$, and $\beta=1,\ldots,m$.
\begin{enumerate}
\item $\Xi=x_j\left(x_i\frac{\partial}{\partial x_i}+u_\gamma\frac{\partial}{\partial u_
\gamma}\right)$ (sum over repeated indices). 

Its second order prolongation is
\begin{equation}
\begin{aligned}
\Xi^{(2)}&=x_jx_i \frac{\partial}{\partial x_i}
+x_ju_\gamma \frac{\partial}{\partial u_\gamma}+\delta_{j,k}\left(u_{\gamma}-x_m 
u_{\gamma,m}\right)\frac{\partial}{\partial u_{\gamma,k}}\\
&-\left(\delta_{j,k}u_{\gamma,m}u_{\gamma,\ell m}+\delta_{j,\ell}x_m u_{\gamma,k m}+x_j 
u_{\gamma,k\ell}\right)\frac{\partial}{\partial u_{\gamma,k\ell}}.
\end{aligned}
\end{equation}
It is
\begin{equation}
\begin{aligned}
&\Xi ^{(2)}\left(\Delta_{\alpha,\alpha+1;1,1,p,q}\right)=\\ 
\;&=\delta_{j,p}x_m \left(u_{\alpha,m q}u_{\alpha+1,11}-
u_{\alpha,11}u_{\alpha+1,mq}\right)\\
\;&+\delta_{j,q}x_m \left(u_{\alpha,m p}u_{\alpha+1,11}-
u_{\alpha,11}u_{\alpha+1,mp}\right)\\
\;&+2\delta_{j,1}x_m \left(u_{\alpha,pq}u_{\alpha+1,1m}-
u_{\alpha,1m}u_{\alpha+1,pq}\right)\\
\;&+2x_j \left(u_{\alpha,pq}u_{\alpha+1,11}-
u_{\alpha,11}u_{\alpha+1,pq}\right)\\
\;&=-x_m(\delta_{j,p}\Delta_{\alpha,\alpha+1;1,1,m,q}\delta_{j,q}+
\Delta_{\alpha,\alpha+1;1,1,m,p}
+2\delta_{j,1}\Delta_{\alpha,\alpha+1;1,m,p,q})\\
\;&-2x_j\Delta_{\alpha,\alpha+1;1,1,p,q},
\end{aligned}
\end{equation}
that vanishes on $\boldsymbol\Delta=\mathbf{0}$;

\item $\Xi=u_\beta\left(x_i\frac{\partial}{\partial x_i}+u_\gamma\frac{\partial}{\partial 
u_\gamma}\right)$ (sum over repeated indices). 

Its second order prolongation is
\begin{equation}
\begin{aligned}
\Xi^{(2)}&=u_\beta x_i \frac{\partial}{\partial x_i}
+u_\beta u_\gamma \frac{\partial}{\partial u_\gamma}+u_{\beta,k}\left(u_{\gamma}-x_i 
u_{\gamma,i}\right)\frac{\partial}{\partial u_{\gamma,k}}\\
&+\left(u_{\gamma}u_{\beta,k\ell}-u_{\beta}u_{\gamma,k\ell}\right.\\
&-\left.x_i(u_{\gamma,i}u_{\beta,k \ell}+u_{\beta,k}u_{\gamma,i \ell}+u_{\beta,l}
u_{\gamma,i k})\right)\frac{\partial}{\partial u_{\gamma,k\ell}}.
\end{aligned}
\end{equation}
It is
\begin{equation}
\begin{aligned}
&\Xi ^{(2)}\left(\Delta_{\alpha,\alpha+1;1,1,p,q}\right)=\\ 
&\;=(u_{\alpha}-x_i u_{\alpha,i}) \left(u_{\gamma,11}u_{\alpha+1,pq}-
u_{\gamma,pq}u_{\alpha+1,11}\right)\\
&\;+(u_{\alpha+1}-x_i u_{\alpha+1,i}) \left(u_{\alpha,11}u_{\gamma,pq}-
u_{\alpha,pq}u_{\gamma,11}\right)\\
&\;-2u_\gamma\left(u_{\alpha,11}u_{\alpha+1,pq}-
u_{\alpha,pq}u_{\alpha+1,11}\right)\\
&\;-2x_iu_{\gamma,1}\left(u_{\alpha,1i}u_{\alpha+1,pq}-
u_{\alpha,pq}u_{\alpha+1,1i}\right)\\
&\;-x_iu_{\gamma,p}\left(u_{\alpha,11}u_{\alpha+1,iq}-
u_{\alpha,iq}u_{\alpha+1,11}\right)\\
&\;-x_iu_{\gamma,q}\left(u_{\alpha,11}u_{\alpha+1,ip}-
u_{\alpha,ip}u_{\alpha+1,11}\right)\\
&\;=(u_{\alpha}-x_i u_{\alpha,i})\Delta_{\gamma,\alpha+1;1,1,p,q}+(u_{\alpha+1}-x_i 
u_{\alpha+1,i})\Delta_{\alpha,\gamma;1,1,p,q}\\
&\;-2u_{\gamma}\Delta_{\alpha,\alpha+1;1,1,p,q}-2x_iu_{\gamma,1}\Delta_{\alpha,\alpha
+1;1,i,p,q}\\
&\;-x_iu_{\gamma,p}\Delta_{\alpha,\alpha+1;1,1,i,q}-x_iu_{\gamma,q}\Delta_{\alpha,\alpha
+1;1,1,i,p},
\end{aligned}
\end{equation}
that vanishes on $\boldsymbol\Delta=\mathbf{0}$.
\end{enumerate}
\end{proof}

\section{Third order Lie remarkable equations}
\label{sec4}
Third order Lie remarkable equations characterized by the Lie algebra of affine 
transformations,
due to dimensionality considerations, exist  for $n=2$ and $m=1$, where the third 
order jet space and the Lie algebra have the same dimension.
In fact, in \cite{MOV-TMP2007}, the following third order Lie remarkable differential equation has 
been characterized:
\begin{equation}
\label{MOV-equation}
\begin{aligned}
&u_{,11}^3 u_{,222}^2+ u_{,22}^3u_{,111}^2 + 6 u_{,11}u_{,12} u_{,22}
u_{,111}  u_{,222}
- 6 u_{,12}u_{,22}^2u_{,111} u_{,112}  \\
&\;- 6 u_{,11} u_{,22}^2u_{,111} u_{,122}  - 6 u_{,11}^2 u_{,12} u_{,122}
u_{,222} - 6 u_{,11}^2 u_{,22}u_{,112} u_{,222}\\
&\;- 8 u_{,12}^3u_{,111} u_{,222} + 9 u_{,11} u_{,22}^2u_{,112}^2 
+ 9 u_{,11}^2 u_{,22}u_{,122}^2 + 12 u_{,12}^2u_{,22}u_{,111} u_{,122} \\
&\;+ 12 u_{,11} u_{,12}^2u_{,112} u_{,222} -18 u_{,11} u_{,12}u_{,22}u_{,112}  u_{,
122}  =0.
\end{aligned}
\end{equation}
Equation \eqref{MOV-equation} admits also the symmetries generated by
\begin{equation}
a\left(x_1\frac{\partial}{\partial x_1}+x_2\frac{\partial}{\partial x_2}+u\frac{\partial}
{\partial u}\right), \qquad a\in\{x_1,x_2,u\},
\end{equation}
\emph{i.e.}, it is invariant with respect to the Lie algebra of projective 
transformations of $\mathbb{R}^3$.

For $n=3$ and $m=1$, the dimension of Lie algebra of affine transformations is 20, 
whereas the dimension of third order jet space is 23, so that we can not characterize a third order Lie remarkable equation. Even if we take the Lie algebra of projective transformations (with dimension 24), neither strongly nor weakly 
Lie remarkable equations exist, since the rank of the distribution generated by 
the third order prolongation of vector fields of projective transformations is 20.

For $n=2$ the dimension of the Lie algebra of projective transformations, $d=(m+2)(m+4)$, 
is greater than the dimension of the jet space, $j=10m+2$  for all values of $m>0$, so 
that we may try to construct third order strongly Lie remarkable equations involving two 
independent variables and an arbitrary number of dependent ones. Thus, in principle,
we may conjecture that a hierarchy of Lie remarkable third order systems 
involving the unknowns $u_\alpha(x_1,x_2)$ ($\alpha=1,\ldots,m$, $m\ge 2$) may exist.

Taking $m=2$, the computation can be easily done, and the following result can be stated.

\begin{theorem}
The third order system of partial differential equations
\begin{equation}
\label{strong223}
\begin{aligned}
&(u_{1,22}u_{2,12} - u_{1,12}u_{2,22})((u_{1,22}u_{2,11} - u_{1,12}u_{2,12})^2\\
&\quad-(u_{1,11}u_{1,22}-u_{1,12}^2)(u_{2,11}u_{2,22}-u_{2,12}^2))u_{1,111} \\
&\;+3(u_{1,11}u_{2,12}-u_{1,12}u_{2,11})(u_{1,12}u_{2,12}-u_{1,22}u_{2,11})\times\\
&\quad\times(u_{1,12}u_{2,22}-u_{1,22}u_{2,12})u_{1,112} \\
&\;+3(u_{1,11}u_{2,12}-u_{1,12}u_{2,11})^2(u_{1,22}u_{2,12} - u_{1,12}u_{2,22})u_{1,122}\\
&\;+(u_{1,11}u_{2,12}-u_{1,12}u_{2,11})^2(u_{1,11}u_{2,22}-u_{1,12}u_{2,12})u_{1,222} \\
&\;+(u_{1,11}u_{1,22}-u_{1,12}^2)(u_{1,22}u_{2,11} - u_{1,11}u_{2,22})
(u_{1,12}u_{2,22}-u_{1,22}u_{2,12})u_{2,111} \\
&\;+3(u_{1,12}^2 - u_{1,11}u_{1,22})(u_{1,12}u_{2,11} - u_{1,11}u_{2,12})
(u_{1,12}u_{2,22}-u_{1,22}u_{2,12})u_{2,112} \\
&\;+(u_{1,12}^2 - u_{1,11}u_{1,22})(u_{1,12}u_{2,11} - u_{1,11}u_{2,12})^2u_{2,222}=0,\\
&3(u_{1,22}u_{2,12} - u_{1,12}u_{2,22})^2(u_{2,11}u_{2,22}-u_{2,12}^2)u_{1,112}\\
&\;+3(u_{1,11}u_{2,22}-u_{1,22}u_{2,11})(u_{1,22}u_{2,12} - u_{1,12}u_{2,22})(u_{2,11}u_{2,22}-u_{2,12}^2)u_{1,122} \\  
&\;+(u_{2,11}u_{2,22}-u_{2,12}^2)((u_{1,11}u_{2,22}-u_{1,22}u_{2,11})^2 \\
&\;+(u_{1,12}u_{2,22}-u_{1,22}u_{2,12})(u_{1,12}u_{2,11}-u_{1,11}u_{2,12}))u_{1,222}\\    
&\;+(u_{1,22}u_{2,12} - u_{1,12}u_{2,22})^3u_{2,111} \\
&\;+3(u_{1,12}u_{2,12}-u_{1,22}u_{2,11})(u_{1,22}u_{2,12} - u_{1,12}u_{2,22})^2u_{2,112} \\
&\;+3(u_{1,22}u_{2,12} - u_{1,12}u_{2,22})((u_{1,22}u_{2,11}-u_{1,12}u_{2,12})^2 \\
&\;- (u_{1,11}u_{1,22}-u_{1,12}^2)(u_{2,11}u_{2,22}-u_{2,12}^2))u_{2,122} \\
&\;+((u_{1,12}u_{2,12}-u_{1,22}u_{2,11})^3 -(u_{1,11}u_{2,22}+u_{1,12}u_{2,12}-2u_{1,22}u_{2,11})\times \\
&\quad\times(u_{1,11}u_{1,22}-u_{1,12}^2)(u_{2,11}u_{2,22}-u_{2,12}^2))u_{2,222}=0,
\end{aligned}
\end{equation}
characterizing a 20--dimensional submanifold in $J^3(\mathbb{R}^4,2)$, is strongly Lie 
remarkable with respect to the Lie algebra of projective transformations of $\mathbb{R}^4$.
\end{theorem}
\begin{proof}
In fact, the rank of the distribution of the third order prolongations of the vector fields of Lie 
algebra of projective transformations of $\mathbb{R}^4$ is maximal (and equal to 
the dimension of third order jet space). Moreover, the rank lowers to 20 when evaluated on the 
submanifold described by \eqref{strong223}, and this completes the proof.
\end{proof}

For higher values of $m$,  the amount of computation required to look for possible third 
order strongly Lie remarkable equations is rapidly increasing. Work in this direction is 
in progress.

\section{Conclusions} 
\label{sec5}
In this paper, we derived some new second and third order  partial differential 
equations uniquely determined by the Lie algebra of affine or projective 
transformations of $\mathbb{R}^{n+m}$. The characterization of second order partial differential equations is done by using only the Lie algebra of affine transformations; however, we proved that these equations are also invariant with respect to the projective transformations. Affine transformations are sufficient to characterize a third order scalar strongly Lie remarkable equation in two independent variables, 
whereas for a system of two third order differential equations in two 
independent variables we need all the symmetries in the Lie algebra of projective 
transformations. 

As far as the second order strongly Lie remarkable differential equations here identified 
are concerned, by using the property according to which they are in one--to--one 
correspondence with the Lie algebra $\mathcal{P}(\mathbb{R}^{n+m})$ of projective transformations, as the classical Monge--Amp\`ere equation for 
a surface in $\mathbb{R}^3$ with vanishing Gaussian curvature, we can consider them as 
the instances of a hierarchy of second order Monge--Amp\`ere equations in 
$\mathbb{R}^{n+m}$. Recalling also that the simplest second order scalar ordinary differential equation is invariant with respect to the 8--dimensional Lie algebra of projective transformations of the plane, the hierarchy of second order differential equations reads: 
\begin{equation}
\begin{aligned}
&\frac{d^2u}{dx^2}=0, \qquad &&\mathcal{P}(\mathbb{R}^2),\\
&\frac{\partial^2 u}{\partial x_1^2}\frac{\partial^2 u}{\partial x_2^2}-
\left(\frac{\partial^2 u}
{\partial x_1\partial x_2}\right)^2=0, \qquad &&\mathcal{P}(\mathbb{R}^3),\\
&\det\left(\left\|\frac{\partial^2 u}{\partial x_i\partial x_j}\right\|\right)=0, \qquad 
&&\mathcal{P}(\mathbb{R}^{n+1}),\\
&\frac{\partial^2 u_\alpha}{\partial x_1^2}\frac{\partial^2 u_{\alpha+1}}{\partial x_p 
\partial x_q}-
\frac{\partial^2 u_{\alpha+1}}{\partial x_1^2}\frac{\partial^2 u_\alpha}{\partial x_p 
\partial x_q}=0
,\qquad && \mathcal{P}(\mathbb{R}^{n+m}),
\end{aligned}
\end{equation}
where $\alpha=1,\ldots,m-1$ and $i,j,p,q=1,\ldots, n$.

Lie remarkable systems of partial differential equations determined by other relevant Lie 
algebras of point symmetries are currently investigated and will be the object of a 
forthcoming paper.

\section*{Acknowledgments} This research has been partly supported by G.N.F.M. of 
``Istituto Nazionale di Alta Matematica''.

\end{document}